\documentclass[runningheads]{llncs}
\usepackage[T1]{fontenc}
\usepackage{graphicx}
\usepackage[margin=2cm]{geometry}
\usepackage[parfill]{parskip}    				
\usepackage[utf8]{inputenc}
\usepackage[english]{babel} 
\usepackage{amsmath,amssymb,bm}
\usepackage{mathrsfs}
\usepackage{url}	
\usepackage{esint}
\usepackage[colorlinks]{hyperref}
\usepackage{enumerate}
\usepackage{outlines}[enumerate]
\usepackage{framed,comment,enumerate}
\usepackage{upgreek}
\usepackage{pgfplots}
\usepackage{float}
\usepackage{tikz,tikz-cd}
\usepackage{rotating}
\usepackage{caption}
\usepackage{multicol}
\usepackage{cancel}
\usepackage{subcaption}
\usepackage[title]{appendix}
\usepackage{tikz-cd}
\usepackage{pgf, pgffor}
\usepackage{subcaption}
\usepackage{stmaryrd}
\usepackage[color=yellow]{todonotes}
\usetikzlibrary{decorations.pathmorphing,arrows.meta,decorations.markings}

\DeclareFontFamily{U}{MnSymbolC}{}
\DeclareSymbolFont{MnSyC}{U}{MnSymbolC}{m}{n}
\DeclareFontShape{U}{MnSymbolC}{m}{n}{
    <-6>  MnSymbolC5
   <6-7>  MnSymbolC6
   <7-8>  MnSymbolC7
   <8-9>  MnSymbolC8
   <9-10> MnSymbolC9
  <10-12> MnSymbolC10
  <12->   MnSymbolC12}{}
\DeclareMathSymbol{\intprod}{\mathbin}{MnSyC}{'270}

\extrafloats{100}
\numberwithin{equation}{section}

\makeatletter
\@tfor\next:=abcdefghijklmnopqrstuvwxyzABCDEFGHIJKLMNOPQRSTUVWXYZ\do{
\def\command@factory#1{%
\expandafter\def\csname b#1\endcsname{\mathbf{#1}}
\expandafter\def\csname fk#1\endcsname{\mathfrak{#1}}
\expandafter\def\csname bb#1\endcsname{\mathbb{#1}}
\expandafter\def\csname cl#1\endcsname{\mathcal{#1}}
\expandafter\def\csname bcl#1\endcsname{\mathbfcal{#1}}
}
\expandafter\command@factory\next
}

\newcommand{\ob}[1]{\overline{#1}}
\newcommand{\wt}[1]{\widetilde{#1}}

\newcommand{\mc}[1]{\mathcal{#1}}

\newcommand{\mcal}[1]{\mc{#1}}

\newcommand{\ad}{\operatorname{ad}}
\newcommand{\wick}{\, \widehat{\diamond} \,}
\newcommand{\Ad}{\operatorname{Ad}}
\newcommand{\dd}{\mathrm{d}}

\def\dt{\mathrm{d}t}

\newtheorem{assumption}[theorem]{Assumption}

\def\p{{\partial}}

\newcommand{\rmd}{{\rm d}}

\makeatother

\def\p{\partial}

\pgfplotsset{compat=1.16}

\begin{document}
\title{Lagrangian averaging of singular stochastic actions for fluid dynamics}
\author{Theo Diamantakis\inst{1}\orcidID{0009-0000-9338-1620} \and
Ruiao Hu\inst{2}\orcidID{0000-0002-4843-1737}}
\institute{Department of mathematics, Imperial College London, UK
\email{t.diamantakis20@imperial.ac.uk} \\ \and
Department of mathematics, Imperial College London, UK
\email{ruiao.hu15@imperial.ac.uk}}
\maketitle              % typeset the header of the contribution
\begin{abstract}
We construct sub-grid scale models of incompressible fluids by considering expectations of semi-martingale Lagrangian particle trajectories. Our construction is based on the Lagrangian decomposition of flow maps into mean and fluctuation parts, and it is separated into the following steps. First, through Magnus expansion, the fluid velocity field is expressed in terms of fluctuation vector fields whose dynamics are assumed to be stochastic. Second, we use Malliavin calculus to give a regularised interpretation of the product of white noise when inserting the stochastic velocity field into the Lagrangian for Euler's fluid. Lastly, we consider closures of the mean velocity by making stochastic analogues of Talyor's frozen-in turbulence hypothesis to derive a version of the anisotropic Lagrangian averaged Euler equation. 

\keywords{Stochastic geometric mechanics \and Generalised Lagrangian mean \and Variational principles}
\end{abstract}
\section{Introduction}
The modelling of multi-scale fluid flows is a well-known problem in fluid dynamics. In previously established work, fluid flows exhibiting multiple spatial scales have typically been modelled by considering averaging techniques and closures to capture the effects of the unresolvable scales on large-scale dynamics. These models are typically referred to as sub-grid scale models. One such model is the Generalised Lagrangian Mean (GLM) model introduced by Andrews \& Mcintyre \cite{andrews_mcintyre_1978} that considered the application of an arbitrary averaging operator to the Lagrangian expansion of ideal incompressible fluid equations and it is later closed in \cite{GJAJA1996343}. To preserve the geometric structure of ideal fluid equations during averaging, one option is to apply averaging to the variational formulation of ideal fluid flows, instead of applying to the equations directly. This approach resulted in several models such as the Lagrangian averaged Navier Stokes model \cite{FHT2001}. More recently, fluid flows that exhibit multiple temporal scales have been modelled by stochastic partial differential equations where the evolution of the fast dynamics is represented by integration against driving stochastic processes such as Brownian motion and geometric rough paths. One class of structure-preserving stochastic perturbation, known as transport noise, is constructed by assuming a semi-martingale decomposition of the Lagrangian fluid flow map $g_t\in \operatorname{Diff}(\mathcal{M})$ (for fluid domain $\mathcal{M}$) in a stochastic variational principle, see, e.g., \cite{Holm2015,DHL2024}. 

Stochastic fluid flows naturally induce a type of ensemble averaging operator. This is the expectation operator $\mathbb{E}(\cdot)$ over the probability space $(\Omega,\mathcal{F},\bbP)$ that supports the driving stochastic processes. $\bbE$ is linear, a projection and commutes with spatial derivative operations. Consequently, the expectation satisfies the required axioms of an average seen in GLM theory \cite{andrews_mcintyre_1978}.
%The expectation operator satisfies the defining properties of the averaging operator. On the fluid domain $\mathcal{M}$, let $f, g : \Omega \times [0,T] \times \mathcal{M} \rightarrow \mathbb{R}$ be smooth $L^2(\mathcal{M})$ functions of $t\in [0,T]$ and $x \in \mathcal{M}$, $a, b\in \mathbb{R}$, we have the following properties.
%\begin{itemize}
    %\item Linearity: $\mathbb{E}[af+bg] = a\mathbb{E}[f] + b\mathbb{E}[g]$.
    %\item Commutativity: 
        %$\mathbb{E}\left[\int_{\mathcal{M}} f dx\right] = \int_\mathcal{M} \mathbb{E}[f]\,dx $ and $\nabla_x \mathbb{E}[f] = \mathbb{E}[\nabla_x f]$.
    %\item Projectivity: $\mathbb{E}[\mathbb{E}[f]] = \mathbb{E}[f]$.
%\end{itemize}
Additionally, one has the analogue of the zero mean fluctuations under expectation operator as martingales. More precisely, for a $\bbR^d$-martingale $Y_t$, the martingale property of It\^o integral implies that $\bbE \left[ \int_0^t \zeta^i_s \dd Y^i_s \right] = 0 \, ,$
% \begin{align*}
%     \bbE \left[ \int_0^t \zeta^i_s \dd Y^i_s \right] = 0 \, , 
% \end{align*}
for any $\clF_t$-adapted vector fields $\zeta_t = (\zeta^i_t)_{i=1}^d \in \mathfrak{X}(\mathcal{M})$.

% In particular, one may consider maps $\Xi^{(\beta)}_t$ satisfying equations $\dd \Xi^{(\beta)}_t = \beta_t (\Xi^{(\beta)}_t) \dt +  \xi_t(\Xi^{(\beta)}_t) \circ \dd W_t$ which modify the drift $u_t \mapsto u_t + \beta$ in \eqref{compmapform}. 

% Our setup shall differ from the assumptions of other GLM type models such as \cite{Holm2002,DarrylDHolm_2002} in that $\bbE \left[ \dd \Xi^{(\beta)}_t \right] = 0$ if and only if $\beta$ is taken as the Stratonovich correction and the map $\Xi_t$ defines a $\clM$-valued martingale. For other choices of $\beta$ this property fails. It is important to note that we only ever consider averages of quantities in either a tangent or linear space, and we do not specify (or define the meaning of) the quantity $\bbE [ \Xi^{(\beta)}_t ]$.

% \todo[inline]{end of GLM section if we do have one.}

% The formal insertion of the semi-martingale $v_t = \dd g_t g_t^{-1}$ satisfying \eqref{vDdefs} into the kinetic energy \eqref{kelagrangian} is ill defined given the nondifferentiability of $g_t$. Indeed, \eqref{vDdefs} implies $v_t$ can only be understood as a function when integrated in time. One may consider a distributional derivative of $\int_0^t v_t \dt$ and observe that $v_t$ contains a $d$-dimensional white noise multiplying the vector fields $\xi = \{\xi^i \}_{i=1}^d$, but this still may not be evaluated into $\ell$ as distributions cannot be multiplied consistently.

In this work, we consider the averaging of stochastic fluid models using the expectation operator to construct probabilistic variants of sub-grid scale models for ideal fluids. Focusing on incompressible flows, the starting assumption is the semi-martingale Lagrangian fluid flow map $g_t \in \operatorname{SDiff}(\mathcal{M})$ similar to the derivation of stochastic transport fluid models in \cite{Holm2015,Memin2014}. 
The main contributions of this work are as follows. In section \ref{sec:expand}, we define the mean velocity by assuming a factorisation of the full flow map into the fluctuation and mean maps. Then, we use Magus expansion of the velocity field of the composite flow such that it is expressed by vector fields and their actions only, in a similar fashion to \cite{MS2003,Shkoller2002}. This way, the expectation operator acts on vector space quantities only and can be easily defined. 
In section \ref{sec:vp}, using tools from Malliavin calculus, we consider a normalised KE Lagrangian functional such that it is well-defined when the input velocity contains white noise terms. To obtain a dynamical equation of the mean flow, we first make dynamical assumptions on the fluctuations vector fields in assumptions \ref{assump:STH} and \ref{assump:orthg}. Then, we explicitly computed the averaged Lagrangian and derived a stochastically averaged version of the anisotropic Lagrangian averaged Euler's equation \cite{MS2003} through an Euler-Poincar\'e variational principle.

% \todo[inline]{Does the expansion work on the equation level? Add link to our previous work of averaging a v.p. and regularising the Lagrangian.\\
% RH: Which equation? \\
% TD: Euler-Poincare or for this case flat Euler eqn\\
% RH: If i average the Euler's equation in Lagrangian form directly, I get
% \begin{align}
%     \frac{d}{dt}g &= u(t, g) + \xi(g) \diamond \dot{W}_t\,,\\
%     \frac{d^2}{dt^2}g &= \p_t u + \frac{d}{dt}g \cdot \diamond \nabla u + \p_t \eta\diamond \dot{W}_t + \frac{d}{dt}g\cdot \diamond \nabla \xi \diamond \dot{W}_t + \xi \diamond \ddot{W}_t
% \end{align}
% which under expectation is just
% \begin{align}
%     \frac{d^2}{dt^2}\bbE[g] = \p_t \bbE[u] + \bbE[u]\cdot\nabla \bbE[u]\,.
% \end{align}
% So not useful unless we choice something more interesting in the noise. TD: Good point. 
% }

\section{Expansion of Kinetic Energy Lagrangian}\label{sec:expand}
Fix a smooth, compact, boundary-less Riemannian manifold $(\clM, \boldsymbol{g})$ of dimension $d$. Throughout this paper we equip $(\clM, \boldsymbol{g})$ with the Levi-Civita connection $\nabla$. Our aim is to derive deterministic dynamical equations for the mean flow velocity of Euler's fluid flow from an averaged variational principle of the form,
\begin{align}
   S :=  \int_0^T \bbE \left[ \ell( v_t) \right] \dt \,, \quad \ell( v_t) := \frac12\int_{\clM} \boldsymbol{g} (v_t, v_t) \dd \mu\,, \quad v_t \in \mathfrak{X}_{\operatorname{div}_\nabla}(\mathcal{M})\,,\label{avglag}
\end{align}
where $\ell$ is the kinetic energy Lagrangian of divergence free vector fields.
% where the kinetic energy Lagrangian is defined by,  
% \begin{align}
%      \quad \ell( v_t) := \frac12\int_{\clM} \boldsymbol{g} (v_t, v_t) \dd \mu \label{kelagrangian}   \, .
% \end{align}
Here, we have denoted $\mu = \operatorname{vol}_{\boldsymbol{g}}$ the Riemannian volume.
The covariant divergence $\operatorname{div}_\nabla$ appearing in \eqref{avglag} and used throughout is defined for any tensor $T \in \mathfrak{T}(\clM)$ by $\operatorname{div}_\nabla (T) := \operatorname{tr}(\nabla T)$ by contracting the covariant $\nabla_i$ index with, by convention, the \emph{first} (possibly after raising) contravariant index of $T$. In the case of vector fields\footnote{One can show that $\operatorname{div}_\nabla$ agrees with the formula $\clL_X \mu = \operatorname{div}_\mu(X) \mu $ when the connection $\nabla$ is Levi-Civita and $\mu$ is the Riemannian volume form. Note that when $\mu = f d^n x$ in local coordinates then $\operatorname{div}_\mu (X) := f^{-1} \partial_{x_i} (f X^i) $.}, $X \in \mathfrak{X}_{\operatorname{div}_\nabla}(\clM)$ implies that
\begin{align*}
    \operatorname{div}_\nabla(X) = \operatorname{tr}\left( d x^i \otimes \frac{\partial}{\partial x^j}\left( \frac{\partial X^j}{\partial x^i} + X^k \Gamma^j_{i k } \right) \right) = \frac{\partial X^i}{\partial x^i} + X^k \Gamma^i_{i k } = 0\,.
\end{align*}
We assume the full Lagrangian flow map $g_t \in \operatorname{SDiff}(\mathcal{M})$ is a semi-martingale satisfying the stochastic differential equation,
\begin{align}
    \dd g_t =  u_t (g_t) \dd t + \sigma_t (g_t) \circ \dd W_t \, , \label{vDdefs}
    % g_t = \int_0^t u_s (g_s) \dd s + \int_0^t \sigma_s (g_s) \circ \dd W_s \, , \label{vDdefs}
\end{align} %\quad \clD_t = \int_0^t \mcal{L}_{u_s} \clD_s \dd s + \int_0^t \mcal{L}_{\xi_s} \clD_s \circ \dd W_s 
such that $v_t$ is formally associated with the vector field $\frac{\dd}{\dt} g_t g^{-1}_t \in \mathfrak{X}_{\operatorname{div}_\nabla}(\mathcal{M})$. Throughout this paper, we denote by $W_t = (W^1_t,\ldots, W^n_t)$ as a $n$-dimensional Brownian motion with identity covariance and abbreviate stochastic differentials such as $\sigma_t(g_t) \circ \dd W_t := \sum_{i=1}^n\sigma^i_t(g_t) \circ \dd W^i_t$. To define the mean flow velocity, we assume the factorisation of the stochastic flow map $g_t \in \operatorname{SDiff}(\clM)$ into 
\begin{align}
    g_t = g^\epsilon_t = \Xi^\epsilon_t \circ \ob{g}_t \,,
\end{align}
where $\ob{g} : [0, T] \rightarrow \operatorname{SDiff}(\mathcal{M})$ is the mean flow map and $\Xi^\epsilon_t$ is a $\epsilon$-parametrised semi-martingale flow map $\Xi^\epsilon_t$, $\Xi:\Omega \times [0,T] \times [0, 1) \rightarrow \operatorname{SDiff}(\clM)$\,,
with following properties
\begin{align*}
    \Xi^0_t = e \quad \forall t \in [0,T]\,, \quad \text{and}\quad  \Xi^\epsilon_0 = e \quad \forall \epsilon \in [0, 1) \,,
\end{align*}
such that $\Xi^\epsilon_t$ is a flow of near identity diffeomorphism. From the mean flow decomposition \eqref{vDdefs}, the canonical definition of the mean flow velocity field is the velocity assumed with the mean flow map $\ob{u}_t := \frac{\dd}{\dt}\ob{g}_t \ob{g}_t^{-1}$. By assumption, $\ob{u}_t$ is a deterministic object and it is natural to seek the dynamics of $\ob{u}_t$ through the variational principle \eqref{avglag} where the stochastic fluctuation induced by the fluctuation map $\Xi^\epsilon_t$ is averaged. Via direct computation, one finds the vector field generated by $g^\epsilon_t$ can be written in terms of the mean velocity $\ob{u}_t$
\begin{align}
    \circ \dd g_t^\epsilon g_t^{\epsilon; -1} = \circ \dd \Xi_t^{\epsilon}\Xi_t^{\epsilon;-1} + \Ad_{\Xi_t^{\epsilon}} \ob{u}_t\,\dt = \sigma_t \circ \dd W_t + u_t \,\dt \,,   \label{eq:stoch composite vf} 
\end{align}
where the notation $\circ \dd g^\epsilon_t$ is there to denote the Stratonovich interpretation of the stochastic differential $\dd g^\epsilon_t$. We observe that $\operatorname{SDiff}(\mc{M})$ does not necessarily possess a vector space structure and the term $\Ad_{\Xi^\epsilon_t}\ob{u}_t$ is non-linear in $\Xi^\epsilon_t$. In order to express terms of the form $\bbE[\Ad_{\Xi^\epsilon_t}\ob{u}_t]$ in the averaged Lagrangian in a workable form, 
we rewrite the fluctuations of $\ob{u}_t$ defined in powers of $\epsilon$ via Taylor expansion as,
\begin{align}
\begin{split}
    \circ\dd g_t^\epsilon g_t^{\epsilon; -1} &= \circ\dd g_t^\epsilon g_t^{\epsilon; -1} \Bigg|_{\epsilon = 0} + \epsilon \frac{\dd}{\dd \epsilon}\Bigg|_{\epsilon = 0} \left( \circ\dd g_t^\epsilon g_t^{\epsilon; -1}  \right) \, + \frac{\epsilon^2}{2}  \frac{\dd^2}{\dd \epsilon^2}\Bigg|_{\epsilon = 0} \left(\circ\dd g_t^\epsilon g_t^{\epsilon; -1} \right)\,  + \clO(\epsilon^3) \\
    & := \circ\dd v^0 + \circ\epsilon \dd v^\prime + \circ\frac{\epsilon^2}{2} \dd v^{\prime \prime} + \clO(\epsilon^3) \,,
\end{split} \label{eq:dgginv expansion}
\end{align}
such that expectation is applied to $\dd g_t^\epsilon g_t^{\epsilon; -1} \in \mathfrak{X}(\mc{M})$. The (formal) Lie algebra of divergence free vector fields is a linear space where these operations can be performed unambiguously. In the expansion \eqref{eq:dgginv expansion}, we will express the first and second order fluctuation fields $\circ \dd v^\prime$ and $\circ \dd v^{\prime\prime}$ from the deformations of $\Xi^\epsilon_t$
\begin{align}
    \frac{\dd}{\dd \epsilon}\Bigg|_{\epsilon = 0} \Xi_t^\epsilon  =: w_t, \quad \frac{\dd^2}{\dd \epsilon^2}\Bigg|_{\epsilon = 0} \Xi_t^\epsilon  =: \chi_t\,, \quad w_t, \chi_t \in \mathfrak{X}_{\operatorname{div}_{\nabla}}(\clM)\,, \quad \forall t \in [0,T]\,.\label{deformations}
\end{align}
Here, the deformation vector fields $w_t$ and $\chi_t$ are assumed to be semi-martingales which possess a unique decomposition into finite variation and (local) martingale parts.
\begin{lemma}
    In terms of the deformation vector fields $w_t$ and $\chi_t$, the fluctuation stochastic vector fields $\dd v^0$, $\dd v^\prime$ and $\dd v^{\prime \prime}$ can be expressed as
    \begin{align}
        \circ \dd v^0 &= \ob{u}_t\,\dt\,, \label{eq:0nd_order_term} \\
        \circ \dd v^\prime &= \circ \dd w_t + \ad_{w_t} \ob{u}_t \,\dt\,, \label{eq:1nd_order_term} \\
        \circ \dd v^{\prime\prime} &= \ad_{w_t} \ad_{w_t} \ob{u}_t\,\dt + \ad_{\chi_t - \nabla_{w_t} {w_t}} \ob{u}_t\,\dt  + \circ \dd \left( \chi_t - \nabla_{w_t} {w_t} \right) + \ad_{w_t} \circ \dd {w_t} \,,\label{eq:2nd order term}
    \end{align}
    where $\nabla$ is a fixed torsion-free connection. For convenience we will use the Levi-Civita connection induced by the metric $\boldsymbol{g}$. 
\end{lemma}
\begin{proof} The $\mathcal{O}(1)$ term $\rmd v^0$ can be deduced immediately from the the near identity property $\Xi^\epsilon_0 \equiv e$ and $\Ad_e \ob{u}_t = \ob{u}_t$,
\begin{align*}
        \rmd v^0 = \dd g_t^0 g_t^{0; -1} = \Ad_{\Xi_t^0} \ob{u}_t\,\dt + \circ \dd \Xi_t^0 \Xi_t^{0;-1} = \ob{u}_t \,\dt \,.
\end{align*}
To calculate higher order terms, we differentiate the decomposition, where the first term is a direct application of the definition of the adjoint representation (Lie chain rule), %and make use of the fact that $\frac{\partial}{\partial \epsilon} $ and $\circ \rmd $ commute. Introducing the shorthand $(\cdot)^\prime$ for $\frac{\rmd }{\rmd \epsilon} (\cdot)$, we have 
\begin{align*} 
    \circ \rmd v^\prime = \frac{\dd}{\dd \epsilon}\Bigg|_{\epsilon = 0} \circ \rmd v^\epsilon &= \frac{\dd}{\dd \epsilon}\Bigg|_{\epsilon = 0} \left(  \Ad_{\Xi_t^\epsilon} \ob{u}_t\,\dt + \circ \dd \Xi_t^\epsilon \Xi_t^{\epsilon; -1} \right), \quad  \frac{\dd}{\dd \epsilon}\Bigg|_{\epsilon = 0} \Ad_{\Xi_t^\epsilon} \ob{u}_t := \ad_{w_t} \ob{u}_t\, .
\end{align*} 
To take the derivative of $\circ \dd \Xi_t^\epsilon \Xi_t^{\epsilon; -1}$, we express $\Xi_t^\epsilon = \exp(\Omega_t^\epsilon)$ and differentiate $\Omega^\epsilon_t$ through the chain rule on the exponential map. This procedure is known as a Magnus expansion \cite{https://doi.org/10.1002/cpa.3160070404}. Although we work formally, the application of this type of theory to the infinite dimensional group of volume preserving diffeomorphisms can be found in \cite{ACC2014,CCR2023,CC2007}. We state the following formulae that can be obtained through known derivatives of the exponential map at the identity \cite{Hall_2003},
\begin{align}
    \circ \dd \Xi^\epsilon_t \Xi^{\epsilon;-1}_t = \sum_{k=0}^\infty \frac{1}{(k+1)!} \ad^k_{\Omega_t^\epsilon} \circ \dd \Omega_t^\epsilon, \quad \ad^k_{\Omega_t^\epsilon}\left( \cdot \right) := \ad_{\Omega_t^\epsilon} \left( \ad^{k-1}_{\Omega_t^\epsilon} \left( \cdot \right) \right) , \quad \ad^0_{\Omega_t^\epsilon} := I \, . \label{magnusvelo}
\end{align}
The time derivative $\circ \dd \Omega^\epsilon_t$ is expressible as the series \cite{BLANES2009151}
\begin{align}
    \circ \dd \Omega^\epsilon_t = \sum_{k=0}^\infty \frac{B_k}{k!} \ad^k_{\Omega_t^\epsilon} \left( \circ \dd \Xi_t^\epsilon \Xi^{\epsilon;-1}_t \right) \, , \label{magnusseries}
\end{align}
where $B_k$ denotes the Bernoulli numbers with $B_1 = -\frac12$. From the initial assumption $\Xi^0_t \equiv e$ it must be the case that $\Omega_t^0 \equiv \circ \dd \Omega^0_t \equiv \circ \dd \Xi_t^0 \equiv 0$ from differentiating a constant.

Using the Magnus generator $\Omega_t^\epsilon$, we can now compute first $\epsilon$-derivatives of the velocity $\dd \Xi^\epsilon_t \Xi^{\epsilon;-1}_t$ by termwise differentiation of the power series \eqref{magnusvelo},
\begin{align}
    \frac{\partial}{\partial \epsilon} \Bigg\vert_{\epsilon = 0} \circ \dd \Xi^\epsilon_t \Xi^{\epsilon;-1}_t = \sum_{k=0}^\infty \frac{1}{(k+1)!} \frac{\partial}{\partial \epsilon} \Bigg\vert_{\epsilon = 0} \left( \ad^k_{\Omega_t^\epsilon} \circ \dd \Omega_t^\epsilon \right) = \frac{\partial}{\partial \epsilon} \Bigg\vert_{\epsilon = 0} \circ \dd \Omega_t^\epsilon =: \circ \dd \delta \Omega_t  \,. \label{magnuso1}
\end{align}
Note that only the $k=0$ term contributes due to the fact $\Omega^0_t = \circ \dd \Omega_t^0 = 0$. We may relate $\delta \Omega_t$ to the first order fluctuation $w_t$ using the $k=0$ term of the Magnus series \eqref{magnusseries},
\begin{align}
    \circ \dd \delta \Omega_t = \delta \circ\dd \Omega_t = \frac{\partial}{\partial \epsilon} \Bigg\vert_{\epsilon = 0} \left( \circ \dd \Xi^\epsilon_t \Xi^{\epsilon;-1}_t \right) = \circ \dd \left( \delta \Xi_t \Xi^{0;-1}_t \right) + {\circ \dd \Xi^0_t \delta ( \Xi^{-1}_t )} = \circ \dd w_t, \quad \delta \Omega_t = \delta \Xi_t := w_t\, . \label{concreteo1}    
\end{align}
Consequently, we recover $\circ \dd v^0 = \circ \dd w_t + \ad_{w_t} \ob{u}_t \,\dt$.

In order to compute the second order fluctuation of $\rmd v^\epsilon$ we require the connection to take second derivatives. We first differentiate $\Ad_{\Xi^\epsilon} \ob{u}$ through twice application of the Lie chain rule, where the shorthand $(\cdot)^\prime$ denotes $\frac{\partial }{\partial \epsilon} (\cdot)$, 
\begin{align*} 
    \frac{\partial^2 }{\partial \epsilon^2 } \Bigg\vert_{\epsilon = 0} \Ad_{\Xi_t^\epsilon} \ob{u} &= \frac{\partial }{\partial \epsilon } \Bigg\vert_{\epsilon = 0}  \ad_{\Xi_t^{\epsilon; \prime}\Xi_t^{\epsilon; -1} } \Ad_{\Xi_t^\varepsilon} \ob{u}_t\,\dt \\
    &= \left( \ad_{ \zeta^\epsilon_t  } \Ad_{\Xi_t^\varepsilon} \ob{u}_t \,\dt + \ad_{\Xi_t^{\epsilon; \prime}\Xi_t^{\epsilon; -1} }\ad_{\Xi_t^{\epsilon; \prime}\Xi_t^{\epsilon; -1} } \Ad_{\Xi_t^\varepsilon} \ob{u}_t \,\dt \right) \Bigg\vert_{\epsilon = 0} \\
    &= \ad_{\chi_t - \nabla_{w_t} w_t} \ob{u}_t \,\dt + \ad_{w_t} \ad_{w_t} \ob{u}_t \,\dt \, , 
\end{align*} 
where $ \zeta^\epsilon_t := \Xi_t^{\epsilon; \prime\prime}\Xi_t^{\epsilon; -1} -\nabla_{\Xi_t^{\epsilon; \prime}\Xi_t^{\epsilon; -1} } \Xi_t^{\epsilon; \prime}\Xi_t^{\epsilon; -1}$. 
In particular, the term $\nabla_{w_t} w_t$ arises when one attempts to compute a subsequent variation of $\Xi_t^{\epsilon; \prime}\Xi_t^{\epsilon; -1}$. With the Magnus expansion, one can write this term as $\delta^2 \Omega_t$ without use of a connection, however a connection is required to relate this quantity to the fluctuations $\chi_t, w_t$, see the Remark \ref{rmknabla} below.

For the second derivative of $\dd \Xi^\epsilon_t \Xi^{\epsilon;-1}_t$ it suffices to check the $k=0,1$ terms in the series \eqref{magnusseries},
\begin{align}
    \frac{\partial^2}{\partial \epsilon^2} \Bigg\vert_{\epsilon = 0} \circ \dd \Xi^\epsilon_t \Xi^{\epsilon;-1}_t = \sum_{k=0}^\infty \frac{1}{(k+1)!} \frac{\partial^2}{\partial \epsilon^2} \Bigg\vert_{\epsilon = 0} \left( \ad^k_{\Omega_t^\epsilon} \circ \dd \Omega_t^\epsilon \right) = \frac{\partial^2}{\partial^2 \epsilon} \Bigg\vert_{\epsilon = 0} \left( \circ \dd \Omega_t^\epsilon + \frac12 \ad_{\Omega_t^\epsilon} \circ \dd \Omega_t^\epsilon\right) = \circ \dd \delta^2 \Omega_t + \ad_{\delta \Omega_t} \circ \dd \delta \Omega_t \label{magnuso2}  \, . 
\end{align}
Using the relation $\delta \Omega_t = w_t$ and $\delta^2 \Omega_t = \chi_t - \nabla_{w_t} w_t$ (proven below), we thus obtain,
\begin{align*}
    \circ \rmd v^{\prime\prime} = \ad_{w_t} \ad_{w_t} \ob{u}_t\,\dt + \ad_{\chi_t - \nabla_{w_t} w_t} \ob{u}_t\,\dt + \circ \dd \left( \chi_t - \nabla_{w_t} w_t \right) + \ad_{w_t} \circ \dd w_t \, ,    
\end{align*}
as was originally claimed. 
\end{proof} 

\begin{remark}[Connection dependence and matrix Lie groups]\label{rmknabla} 
The expressions \eqref{magnuso1}, \eqref{magnuso2} define equations on any Lie algebra in the absence of any additional structure. As a consequence of $\circ \dd \Xi^0_t = 0$, \eqref{concreteo1} can always be used to relate the variations of $\Omega$ and $\Xi$. The second order term $\delta^2 \Omega_t$ differs, as we must define an action of Lie algebra terms using external structures such as a connection. In order to proceed, one formally writes,
\begin{align} 
\begin{split}
    \circ \dd \delta^2 \Omega_t &= \delta^2 \circ \dd \Omega_t = \frac{\partial^2}{\partial \epsilon^2} \Bigg\vert_{\epsilon = 0} \left( \circ \dd \Xi^\epsilon_t \Xi^{\epsilon;-1}_t -\frac12 \ad_{\Omega_t^\epsilon} \circ \dd \Xi^\epsilon_t \Xi^{\epsilon;-1}_t \right) \\
    &= \frac{\partial}{\partial \epsilon} \Bigg\vert_{\epsilon = 0} \left( \circ \dd \left[ \frac{\partial \Xi^\epsilon_t}{\partial \epsilon} \Xi^{\epsilon;-1}_t \right]  + \ad_{\frac{\partial \Xi^\epsilon_t}{\partial \epsilon} \Xi^{\epsilon;-1}_t} \circ \dd \Xi^\epsilon_t \Xi^{\epsilon;-1}_t \right) - \ad_{\delta \Omega_t} \delta (\circ \dd \Xi_t \Xi^{-1}_t ) \\
    &= \circ \dd \left( \delta^2 \Xi_t\Xi^{0;-1}_t + \delta \Xi_t \delta \Xi^{-1}_t \right) + \ad_{\delta \Xi_t} \delta ( \dd \Xi_t \Xi_t^{-1}) - \ad_{\delta \Omega_t} \delta (\dd \Xi_t \Xi^{-1}_t ) = \circ \dd \left( \chi_t + \delta \Xi_t \delta \Xi^{-1}_t \right)
    \, .
\end{split} \label{concreteo2}
\end{align}
When $\fkg$ is finite dimensional, one may appeal to Ado's Theorem to embed $\fkg \subset \operatorname{GL}(N)$ for some $N \in \bbN$ and give meaning to $\delta \Xi_t \delta \Xi^{-1}_t = -\delta \Xi_t \Xi_t^{0;-1} \delta \Xi_t \Xi_t^{0;-1} = -\delta \Xi_t \delta \Xi_t = -w^2$ as a product of matrices. 
%\begin{align*}
    %\frac{\partial^2}{\partial \epsilon^2} \Bigg\vert_{\epsilon = 0} \circ \dd \Xi^\epsilon_t \Xi^{\epsilon;-1}_t = \circ \dd \delta^2 \Omega_t + \ad_{\delta \Omega_t} \circ \dd \delta \Omega_t = \circ \dd \left( \chi_t - w_t^2 \right)  + \ad_{w_t} \circ \dd w_t \, %.
%\end{align*}
One can then show the second order expansion cancels significantly as,
\begin{align*} 
\begin{split}
\circ \rmd v^{\prime\prime} &= \ad_{w_t} \ad_{w_t} \ob{u}_t\,\dt + \ad_{\chi_t - w_t^2} \ob{u}\,\dt + \circ \dd \left( \chi_t - w^2 \right) + \ad_{w_t} \circ \dd w_t = \circ \dd \chi_t + \ad_{\chi_t} \ob{u}_t \dt - 2 w_t ( \circ \dd v^\prime) \, ,
\end{split}
\end{align*}
under the property that $\ad$ reduces to the commutator of $\operatorname{GL}(N)$ matrices within this embedding. In the case of the diffeomorphism group, a concrete representation of \eqref{magnuso2} in terms of $\delta \Xi_t$ requires identifying $\delta \Xi_t \delta \Xi^{-1}_t := -\nabla_{w_t} w_t$ as a covariant derivative. In Euclidean space, we recover consistency with the ordinary chain rule where the connection is canonical, 
\begin{align*}
    \frac{\partial }{\partial \epsilon} \Bigg\vert_{\epsilon = 0} \delta \Xi_t \circ \Xi^{\epsilon;-1}_t &= \left[ D ( \delta \Xi_t) \circ \Xi^{\epsilon;-1}_t  \right] \cdot \frac{\partial \Xi^{\epsilon;-1}_t }{\partial \epsilon} \Bigg\vert_{\epsilon = 0} = -\left[ D ( \delta \Xi_t) \circ \Xi^{\epsilon;-1}_t  \right] \cdot \left[ D ( \Xi^\epsilon_t) \circ \Xi^{\epsilon;-1}_t  \right]^{-1}  \cdot \left( \frac{\partial \Xi^{\epsilon}_t }{\partial \epsilon}  \circ \Xi_t^{\epsilon;-1} \right) \Bigg\vert_{\epsilon = 0} \\
    &= -D ( \delta \Xi_t)  \cdot \delta \Xi_t = - \nabla^E_{\delta \Xi_t} \delta \Xi_t \, .
\end{align*} 
In the above, we have denoted by $\nabla^E$ the Euclidean connection and $D$ the total (Frech\^et) derivative for maps $\mathbb{R}^n \mapsto \mathbb{R}^n$, with the second equality was derived from the inverse function theorem. In curved geometries, the choice of $\nabla$ is non-canonical, although the Levi-Civita connection is a natural choice we take throughout. %Within this work we shall make the natural choice that $\nabla$ is the Levi-Civita connection induced by the choice of metric $\boldsymbol{g}$.  
\end{remark}

\begin{lemma} \label{2ndorderaltlemma}
The second order fluctuation vector fields $\dd v^{\prime \prime}$ may be written equivalently as
\begin{align} 
    \circ \dd v^{\prime\prime} = \circ \dd \chi_t + \ad_{\chi_t} \ob{u}_t\,\dt - \nabla \nabla \ob{u}_t (w_t,w_t)\,\dt + R(\ob{u}_t, w_t) w_t\,\dt  - 2 \nabla_{w_t} \circ \dd v^\prime\,. \label{eq:2nd order term alt}
\end{align}
where the Riemann curvature tensor $R = R^{\beta}_{i j \alpha } d x^i \wedge d x^j \otimes d x^\alpha \otimes \frac{\partial}{\partial x^\beta}$ is a tensor field of type $(3,1)$ and a section of $\Lambda^2 (T^*M) \otimes \operatorname{End}(TM)$.
\end{lemma}
\begin{proof} 
This follows from repeated use of the following identities,
\begin{align}
   \ad_{X}Y = -[X,Y] = -\nabla_X Y + \nabla_Y X, \quad \nabla \nabla Z (X,Y) = \nabla_X \nabla_Y Z - \nabla_{\nabla_X Y} Z, \quad \forall X,Y,Z \in \fkX(\clD) \,, \label{covidentities}
\end{align}
which follows directly from the definition of being a torsion-free connection and the Leibniz property. Here, we have the covariant Hessian operator $\nabla \nabla $ defined on tensor $T \in \mathfrak{T}(\clM)$ of type $(p,q)$ by $\nabla\nabla T = \nabla_i\nabla_j T dx^i\otimes dx^j$. Inserting \eqref{covidentities} into the equation for $\dd v^{\prime \prime}$ thus gives,
\begin{align*} 
\begin{split}
\circ \rmd v^{\prime\prime} &= \ad_{w_t} \ad_{w_t} \ob{u}_t\,\dt + \ad_{\chi_t - \nabla_{w_t} w_t} \ob{u}_t\,\dt + \circ \dd \left( \chi_t - \nabla_{w_t} w_t \right) + \ad_{w_t} \circ \dd w_t \\
&= \ad_{\chi_t - \nabla_{w_t} w_t} \ob{u}_t\,\dt + \circ \dd \left( \chi_t - \nabla_{w_t} w_t \right) + \ad_{w_t} \circ \rmd v^\prime \\
&= \circ \dd \chi_t + \ad_{\chi_t} \ob{u}_t\,\dt - \left(\nabla_{\circ \dd w_t} w_t + \nabla_{w_t} \circ \dd w_t + \ad_{\nabla_{w_t} w_t} \ob{u}_t\,\dt\right) + \ad_{w_t} \circ \rmd v^\prime \\
&= \circ \dd \chi_t + \ad_{\chi_t} \ob{u}_t\,\dt + \nabla_{\ad_{w_t} \ob{u}_t} w_t\,\dt + \nabla_{w_t} (\ad_{w_t} \ob{u}_t)\,\dt - \ad_{\nabla_{w_t} w_t} \ob{u}_t\,\dt - 2 \nabla_{w_t} \circ \rmd v^\prime \\
&= \circ \dd \chi_t + \ad_{\chi_t} \ob{u}_t\,\dt + \nabla_{-\nabla_{w_t} \ob{u}_t + \nabla_{\ob{u}_t} w_t } w_t\,\dt + \nabla_{w_t} ( -\nabla_{w_t} \ob{u}_t + \nabla_{\ob{u}_t} w_t )\,\dt + \nabla_{\nabla_{w_t} w_t} \ob{u}_t\,\dt \\
& \hspace{20em }- \nabla_{\ob{u}_t} \nabla_{w_t} w_t \,\dt - 2 \nabla_{w_t} \circ \rmd v^\prime \\
&= \circ \dd \chi_t + \ad_{\chi_t} \ob{u}_t\,\dt - \nabla \nabla \ob{u}_t (w_t,w_t)\,\dt - \nabla \nabla w_t (\ob{u}_t, w_t)\,\dt + \nabla \nabla w_t (w_t, \ob{u}_t)\,\dt - 2 \nabla_{w_t} \circ \rmd v^\prime \\
&= \circ \dd \chi_t + \ad_{\chi_t} \ob{u}_t\,\dt - \nabla \nabla \ob{u}_t (w_t,w_t)\,\dt + R(\ob{u}_t, w_t) w_t\,\dt  - 2 \nabla_{w_t} \circ \rmd v^\prime\,.
\end{split}
\end{align*}
\end{proof}
In the last line, we picked up a curvature term as a consequence of the identity,
\begin{align*}
    R(X,Y)Z := -\nabla_X \nabla_Y Z + \nabla_Y \nabla_X Z + \nabla_{[X,Y]}Z = -\nabla \nabla Z (X,Y) + \nabla \nabla Z (Y,X)\,.
\end{align*}
The curvature term $ R(\ob{u},w) w $ vanishes in flat geometries, but otherwise may be expressed as $ R(\ob{u}, w) w  = R(\ob{u}, \cdot) : (w\otimes w)$ as a contraction with a $(2,1)$ tensor field $R(\ob{u}, \cdot) = R^{\beta}_{i j \alpha } \ob{u}^i d x^j \otimes d x^\alpha \otimes \frac{\partial}{\partial x^\beta}$ with a $(2,0)$ tensor field $w \otimes w$. 
With the alternative formulation of Lemma \ref{2ndorderaltlemma}, the full expansion of $\circ \dd g^\epsilon_t g_t^{\epsilon;-1}$ to $\mathcal{O}(\epsilon^2)$ is given by,
\begin{align}
\begin{split}
    \circ \dd g^\epsilon_t g_t^{\epsilon;-1} = \ob{u}_t\,\dt & + \epsilon \left( \circ \dd w_t + \ad_{w_t} \ob{u}_t \, \dt\right) \\
    & + \frac{\epsilon^2}{2} \left( \circ \dd \chi_t + \ad_\chi \ob{u}_t\,\dt - \nabla \nabla \ob{u}_t (w_t,w_t)\,\dt + R(\ob{u}_t, w_t) w_t\,\dt  - 2 \nabla_{w_t} \circ \dd v^\prime \right) + \clO( \epsilon^3) \,, 
\end{split}\label{eq:u expansion}
\end{align}

% As we will be considering expectations of the velocity decomposition, the It\^o representation of the Stratonovich process $\circ \dd v^{\prime\prime}$ can be written as
% \begin{align}
%     \dd v^{\prime\prime} = \dd \chi + \ad_\chi \ob{u}\,\dt - \nabla \nabla \ob{u} (w,w)\,\dt + R(\ob{u}, w) w\,\dt  - 2 \nabla_w \dd v^\prime - \dd \left[ \nabla_w, v^\prime \right]_t\,.
% \end{align}
% where $\left[ \cdot, \cdot \right]_t$ is the quadratic covariation. 
% which can be expressed in terms of Wick product $\wick$ and the distributional white noise $\dot{W}^i$ to be
% \begin{align}
%     \dd g^\epsilon_t g_t^{\epsilon;-1} = \ob{u} + \epsilon \left(\p_t w + \ad_{w} \ob{u}\right) + \frac{\epsilon^2}{2} \left( \p_t\chi + \ad_{\chi} \ob{u} - \nabla \nabla \ob{u} (w,w) + R(w_t, \ob{u}) w - 2 \nabla_w v^\prime \right) + \clO( \epsilon^3) \,, 
% \end{align}
\begin{remark}[Expectation of advected quantities]
Advected quantities $\ob{a}_t \in \mathfrak{T}(\clM)$ of the mean flow $\ob{u}_t$ can also be perturbed in similar fashion to the perturbed stochastic velocity field $\circ \dd g^\epsilon g^{\epsilon;-1}$ and expanded in \eqref{eq:u expansion} by considering the quantity $a^\epsilon_t := \Xi^\epsilon_{t*} \ob{a}_t = g^\epsilon_{t*} a_0$. Applying the standard Lie chain rule in the $\epsilon$-parameter, 
\begin{align*}
    \frac{\dd }{\dd \epsilon}\left(\Xi^\epsilon_{t*} a^\epsilon_t\right) = \Xi^\epsilon_{t*} \frac{\dd }{\dd \epsilon} a^\epsilon_t  - \mcal{L}_{\Xi^{\epsilon;\prime}_t \Xi^{\epsilon;-1}_t} a^\epsilon_t\,,
\end{align*}
where $\mcal{L}: \mathfrak{X}(\clM) \times \mathfrak{T}(\clM) \rightarrow \mathfrak{T}(\clM)$ is the Lie derivative operator, we have the expansion 
\begin{align}
    a^\epsilon_t = \ob{a}_t - \epsilon \mcal{L}_{w_t} \ob{a}_t + \frac{\epsilon^2}{2} \left( -\mcal{L}_{\chi_t - \nabla_{w_t} w_t} \ob{a}_t + \mcal{L}_{w_t} \mcal{L}_{w_t} \ob{a}_t \right) + \clO(\epsilon^3)\,.
\end{align}    
\end{remark}

\section{Closure via variational principle}\label{sec:vp}
The direct substitution of the expansion \eqref{eq:u expansion} into the Lagrangian \eqref{avglag}, dropping all cubic and higher order terms, results in the following formal action,
\begin{align} 
\begin{split}
    S = \bbE \Bigg[ &\int_0^T \int_\clM\frac12 \boldsymbol{g}(\ob{u}_t, \ob{u}_t)\,\dt +  \int_0^T \int_\clM\epsilon \boldsymbol{g}(\ob{u}_t, \circ \dd w_t + \ad_{w_t} \ob{u}_t\,\dt ) \\ 
    & \quad +\int_0^T \int_\clM\frac{\epsilon^2}{2} \boldsymbol{g}\left(\ob{u}_t, \circ \dd \chi_t + \ad_{\chi_t} \ob{u}_t\,\dt - \nabla \nabla \ob{u}_t (w_t,w_t)\,\dt + R(\ob{u}_t, w_t) w_t\,\dt  - 2 \nabla_{w_t} \circ \dd v^\prime \right) \\
    & \quad + \int_0^T \int_\clM \frac{\epsilon^2}{2} \boldsymbol{g}\left( \dot{w}_t + \ad_{w_t} \ob{u}_t, \dot{w}_t + \ad_{w_t} \ob{u}_t \right)\,\dt \Bigg] := \bbE[I_1 + I_2 + I_3 + I_4]\,.
\end{split} \label{biglagrangian}
\end{align} 
Here, $I_1$ is the kinetic energy of the mean velocity $\ob{u}_t$; $I_2$ and $I_3$ are mixed products with stochastic fluctuations $w_t$ and $\chi_t$, of the mean velocity $\ob{u}_t$, understood as stochastic integrals. The kinetic energy of the fluctuations $I_4$ is ill-defined due to terms representing the formal product of white noise and has no classical interpretation. 

In existing literature of stochastic variational principles where the action contains explicit integration against semi-martingales, e.g., \cite{CCR2023,DHP2023,Holm2015,HH2021a}, terms such as $I_4$ are typically dropped. Our approach is to instead retain these terms under an average, which allow for a renormalisation of $I_4$ into a meaningful quantity. The method we employ is to interpret $I_4$ using the Wick product, which defines products of white noise through Malliavin calculus.

We very briefly discuss this construction. Since our usage is only elementary, we do not dwell substantially on definitions. For precise details, we refer the reader to \cite{GjessingHoldenLindstrømØksendalUbøeZhang+1993+29+67} and \cite{di_nunno_malliavin_2009}. One defines a \emph{white noise probability space} $(\Omega, \mathcal{B}, \mu)$ of random distributions inside a sample space of $\mathbb{R}^n$ valued tempered distributions $\Omega := \clS^*(\bbR_+; \bbR^n)$. A probability measure and Borel sigma algebra exists on this vector space via the Bochner-Milnos theorem.

One may then define $L^2(\Omega)$-integrable random variables taking values in the Schwartz class which may be characterised by convergent It\^o-Weiner chaos expansions. The elements of $L^2(\Omega)$ with a rapidly decreasing series expansion are known as the Hida test functions $\mathcal{S}(\Omega)$, and elements its dual space $\mathcal{S}^*(\Omega)$ are known as Hida distributions. Within this space, a distribution valued stochastic process $\dot{W}_t$ may be defined. Crucial for us is that $\mathcal{S}^*(\Omega)$ contains a closed, associative, commutative product that is defined by the formal Cauchy product of the series defined in an $L^2(\Omega)$ orthogonal basis,  
\begin{align*}
   S(\omega, t) = \sum_{\substack{\alpha = (\alpha_1, \ldots, \alpha_l) \\[0.3em] l \in \mathbb{N}}}  c_\alpha(t) H_\alpha(\omega), \quad T(\omega, t) = \sum_{\substack{\beta = (\beta_1, \ldots, \beta_m) \\[0.3em] m \in \mathbb{N}}} \wt{c}_\beta(t)  H_\beta(\omega), \quad S \wick T (\omega,t) := \sum_{\gamma = \alpha + \beta} c_{\alpha}(t) \wt{c}_{\beta}(t) H_{\alpha + \beta}(\omega)\, .   
\end{align*}
The Wick product satisfies several desirable properties, firstly, $\mathbb{E} [ S \wick T] = \mathbb{E}[S] \cdot \mathbb{E}[T]$, where $S, T \in \mathcal{S}^*(\Omega)$ need not be independent random variables. Secondly, if either $S, T$ is deterministic (i.e. $S(t) = \bbE [S](t) := c_0(t)$), then $S \wick T = S \cdot T$ reduces to the ordinary product (well defined as a chaos expansion times a constant). Finally, is the correspondence of an It\^o diffusion with its $S^*(\Omega)$ valued ``Wick calculus" representation
\begin{align*}
    \dd X_t = u(X_t) \,\dt + \sigma(X_t)\dd W_t \iff \dot{X}_t = u(X_t) + \sigma(X_t)\wick \dot{W}_t\,,
\end{align*}
 which can be proven rigorously using the Skorohod integral from Malliavin calculus. Interpreting white noise as a Hida distribution value quantity, $\dot{W} \in \mathcal{S}^*(\Omega)$ with $\bbE [\dot{W}_t] = 0$, we \emph{choose} to interpret the white noise products of $I_4$ in the Wick sense. With this choice the expectation distributes to cancel $I_4$. In the other integrals $I_1, I_2, I_3$ we have at least one deterministic factor, making the distinction between Wick and ordinary products trivial.
\begin{comment}
extend the definition of Lagrangian functional such that it can be distribution valued. For the kinetic energy Lagrangian \eqref{kelagrangian}, we define its stochastic generalisation by the Lagrangian $\ell^{\wick}: \mathcal{S}^*(\Omega, \mathfrak{X}(\clM)) \times \mathcal{S}^*(\Omega, \mathfrak{X}(\clM)) \rightarrow \mcal{S}^*(\Omega)$ to be 
\begin{align}
    \ell^{\wick}(v_t) = \frac{1}{2}\int_{\clM} \sum_{i,j=1}^d \boldsymbol{g}_{ij}v_t^i \wick v_t^j d\mu =: \frac{1}{2}\int_{\clM} \boldsymbol{g}^{\wick}(v_t, v_t) d\mu\,,\label{eq:white noise lag}
\end{align}
such that it is well defined when $v_t$ are distribution valued. We remark that when $v_t$ is deterministic, the Wick product Lagrangian reduces to the deterministic case. The distributive property of $\mathbb{E}(\cdot)$ over $\wick$ implies that 
\begin{align}
\begin{split}
    \mathbb{E}\left[\int_0^T\ell^{\wick}(v_t)\right] &= \int_0^T \frac{1}{2}\sum_{i,j=1}^d \boldsymbol{g}_{ij}\mathbb{E}\left[v^i_t \wick v^j_t\right]\,d\mu \\
    &= \int_0^T \frac{1}{2}\sum_{i,j=1}^d \boldsymbol{g}_{ij}\mathbb{E}\left[v^i_t\right]\mathbb{E}\left[v^j_t\right]\,d \mu =  \int_0^T \ell^{\wick}(\mathbb{E}\left[v_t\right]) =  \int_0^T \ell(\mathbb{E}\left[v_t\right])\,, \label{eq:avg action}
\end{split}
\end{align}
which relates the expectation of the action of the Wick production Lagrangian to the deterministic Lagrangian with the expectation of the arguments. 
\end{comment}
Thus, we will proceed by making assumptions to the dynamics of $w_t$ and $\chi_t$ and take averages of all quantities.

%for an arbitrary stochastic process $X_t$, we can equivalently consider action defined by substituting the It\^o form of the velocity expansion \eqref{eq:u expansion} and its expectation before inserting into the general KE Wick product Lagrangian \eqref{eq:avg action}. For explicit computation of the It\^o form of the velocity decomposition, the dynamical equations of $w_t$ and $\chi_t$ are required. In the current setting, we explicitly prescribe deformation diffeomorphism $\Xi^\epsilon_t$ and treat it as input to the action. 
\begin{assumption}[Additive stochastic Taylor hypothesis]\label{assump:STH}\,\\
    We assume that the dynamics of the first order fluctuation vector field $w_t$ is driven by Brownian motion additively such that it can be written as
    \begin{align}
        \circ \dd w_t + \ad_{w_t} \ob{u}_t\,\dt = \xi \circ \dd W_t\,, \label{STH}
    \end{align}
    for some $\xi \in \mathfrak{X}_{\operatorname{div}_\nabla}^n(\clM)$ that are prescribed, divergence free and temporally constant.
\end{assumption}
In this case, the It\^o and Stratonovich interpretations coincide,
\begin{align*}
    \xi \circ \dd W_t := \sum_{i=1}^n \xi_i \circ \dd W^i_t = \sum_{i=1}^n \xi_i \dd W^i_t = \xi \dd W_t\,.
\end{align*}
In particular, taking the expectation of \eqref{STH} implies that $\bbE \left[ w_t \right] =: \ob{w}_t$ is advected by the mean flow $\ob{u}_t$
\begin{align*}
    \p_t \ob{w}_t + \ad_{\ob{w}_t} \ob{u}_t = 0\,,
\end{align*}
which is equivalent for $w_t$ to satisfy Taylor's frozen turbulence hypothesis in expectation. From the additive stochastic Taylor hypothesis \eqref{STH}, via direct calculation we find the dynamics of the fluctuation correlation tensor $F_t := \mathbb{E}[w_t\otimes w_t] \in \mathfrak{X}(\clM)\otimes \mathfrak{X}(\clM)$ as
\begin{align}
    \p_t F_t + \mcal{L}_{\ob{u}_t} F_t = (\xi\otimes \xi) \,, \label{eq:F eq}
\end{align}
in geometric notation. In terms of index notations, it can be written as
\begin{align*}
    \p_t F^{ij}_t + \ob{u}_t^k\nabla_k \nabla F^{ij}_t - F^{kj}_t \nabla_k \ob{u}_t^i - F^{ik}_t \nabla_k \ob{u}_t^j = \xi^i\xi^j \,.
\end{align*}
Under the additive stochastic Taylor Hypothesis, the velocity decomposition up to $\mathcal{O}(\epsilon^2)$ becomes
\begin{align*}
    \circ \dd g^\epsilon_t g_t^{\epsilon;-1} = \ob{u}_t\,\dt &+ \epsilon \xi \circ dW_t + \frac{\epsilon^2}{2} \left( \circ \dd \chi_t + \ad_{\chi_t} \ob{u}_t\,\dt - \nabla \nabla \ob{u}_t (w_t,w_t)\,\dt + R(\ob{u}_t, w_t) w_t\,\dt  - 2 \nabla_{w_t} \xi \circ dW_t \right) \,,
\end{align*}
which can be expressed in Wick calculus and distributional white noise ${\dot{W}}_t$ as
\begin{align}
\begin{split}
    v^\epsilon_t = \p_t g^\epsilon_t g^{\epsilon;-1}_t = \ob{u}_t + \epsilon \xi \wick \dot{W}_t & + \frac{\epsilon^2}{2}\left(\p_t \chi_t + \ad_{\chi_t} \ob{u}_t - \nabla\nabla\ob{u}_t (w_t, w_t) + R( \ob{u}_t, w_t)w_t - 2\nabla_{w_t} \xi \wick\dot{W}_t \right) \\
    &- \frac{\epsilon^3}{2}\left(\nabla_{\xi_i}\nabla_{w_t}\xi_i + \nabla_{\nabla_{w_t} \xi_i}\xi_i + \nabla_{w_t}\nabla_{\xi_i}\xi_i\right)\,\dt\,. \label{eq:white noise v expansion}
\end{split}    
\end{align}
Here, the Stratonovich to It\^o correction for the additive stochastic Taylor hypothesis results in higher-order terms that are dropped in the expansion by $\epsilon$. Noting that there is no canonical choice when applying the expectation, one may apply it to the velocity decomposition expressed on $\p_t g^{\epsilon}_t \in T_{g^\epsilon_t}\operatorname{SDiff}(\mathcal{M})$ or on $v^\epsilon_t = \p_t g^\epsilon_t g^{\epsilon,-1}_t = T_e\operatorname{SDiff}(\mathcal{M}) \simeq \mathfrak{X}(\mathcal{M})$, which corresponds to the Lagrangian and Eulerian expressions of the fluid velocity respectively. 
In taking the expectation of the Euler-Poincar\'e variational principle, we make the modelling choice to apply the expectation directly to the vector field $v^\epsilon_t$ at fixed Eulerian coordinates.
Thus, we obtain the expected mean flow expansion in terms of $\epsilon$, 
\begin{align}
    \ob{v}^\epsilon_t := \mathbb{E}[v^\epsilon_t] = \ob{u}_t + \frac{\epsilon^2}{2}\left(\p_t \ob{\chi}_t + \ad_{\ob{\chi}_t} \ob{u}_t - \nabla \nabla \ob{u}_t : F_t + R(\ob{u}_t, \cdot) : F_t \right) + \mathcal{O}(\epsilon^3)\,.
\end{align}
\begin{assumption}[Expectation orthogonality hypothesis]\,\\ \label{assump:orthg}
    We assume that the \emph{material} derivative of the expectation of the second order fluctuation vector field $\chi_t$ is orthogonal to the mean flow $\ob{u}_t$ under metric pairing. That is
    \begin{align}
        \int_\clM \boldsymbol{g}(\ob{u}_t, \p_t \ob{\chi}_t + \nabla_{\ob{u}_t} \ob{\chi}_t) = 0\,, \label{eq:orthog cond}
    \end{align}
    where $\ob{\chi}_t := \mathbb{E}[\chi_t]$. Noting that $\ob{u}_t$ is assumed to be deterministic and $\bbE[\cdot]$ commutes with spatial derivatives and integrals, the condition \eqref{eq:orthog cond} can be rewritten as 
    \begin{align*}
        \int_\clM \boldsymbol{g}(\ob{u}_t, \circ \dd \chi_t + \nabla_{\ob{u}_t} \chi_t\,\dt) = \dd M_t\,,
    \end{align*}
    where $M_t \in \mathbb{R}$ is a martingale.
\end{assumption}
Under the assumptions \ref{assump:STH} and \ref{assump:orthg}, substituting the white noise velocity expansion \eqref{eq:white noise v expansion} into the white noise Lagrangian \eqref{biglagrangian} we have
\begin{align*}
    \begin{split}
    S &= \bbE \Bigg[ \int_0^T \int_\clM\frac12 \boldsymbol{g}(\ob{u}_t\,,\, \ob{u}_t)\,\dt +  \int_0^T \int_\clM\epsilon \boldsymbol{g}(\ob{u}_t, \xi )\wick \dot{W} \,\dt + \int_0^T \int_\clM \frac{\epsilon^2}{2} \boldsymbol{g}\left( \xi \,,\,\xi \right)\wick \dot{W}_t \wick \dot{W}_t\,\dt\\ 
    & \qquad \qquad +\int_0^T \int_\clM\frac{\epsilon^2}{2} \boldsymbol{g}\left(\ob{u}_t\,,\, \p_t \chi_t + \ad_{\chi_t} \ob{u}_t - \nabla \nabla \ob{u}_t (w_t,w_t) + R(\ob{u}_t, w_t) w_t  - 2 \nabla_{w_t} \xi \wick \dot{W_t} \right)\,\dt  \Bigg] \\
    &= \int_0^T \int_\clM\frac12 \boldsymbol{g}(\ob{u}_t\,,\, \ob{u}_t) + \frac{\epsilon^2}{2} \boldsymbol{g}\left(\ob{u}_t\,,\, \p_t \ob{\chi}_t + \ad_{\ob{\chi}_t} \ob{u}_t - \nabla \nabla \ob{u}_t : F_t + R(\ob{u}_t, \cdot) : F_t \right)\,\dt\\
    &= \int_0^T \int_\clM \frac12 \boldsymbol{g}(\ob{u}_t, \ob{u}_t) + \frac{\epsilon^2}{2} \boldsymbol{g}\left(\ob{u}_t, -\nabla_{\ob{\chi}_t}\ob{u}_t - \nabla \nabla \ob{u}_t : F_t + R(\ob{u}_t, \cdot) : F_t \right)\,\dt\,.
    \end{split}
\end{align*}
Here, we have used the property that $\bbE[ b(X_t) \wick \dot{W}_t] = 0$ which is analogous to the martingale property to It\^o integrals in the second equality and the orthogonality condition \eqref{eq:orthog cond} in the third equality. 
Noting that $\ob{\chi}_t + \operatorname{div}_\nabla F_t = \mathbb{E}\left[\chi_t + \operatorname{div}_\nabla (w_t \otimes w_t)\right] = \mathbb{E}\left[\chi_t + \nabla_{w_t} w_t\right] \in \mathfrak{X}_{\operatorname{div}_\nabla}(\clM)$, through integration by parts and appropriate boundary conditions we have 
\begin{align}
    \begin{split}
        S:= \int_0^T \int_\clM &\frac{1}{2}\boldsymbol{g}(\ob{u}_t, \ob{u}_t) + \frac{\epsilon^2}{2}\boldsymbol{g}(\nabla \ob{u}_t, \nabla \ob{u}_t \cdot F_t) + \frac{\epsilon^2}{2}\boldsymbol{g}_{ij} R^i_{kln}\ob{u}_t^j \ob{u}_t^l F^{kn}_t \,.
    \end{split}\label{eq:incompressible Lag}
\end{align}
Using the Euler-Poincar\'e constrained variations for $\ob{u}_t$, $\delta \ob{u}_t = \p_t \ob{u}_t - \ad_{\eta_t} \ob{u}_t$ for arbitrary $\eta_t \in \mathfrak{X}_{\operatorname{div}_\nabla}(\clM)$ vanishing at the boundaries, Hamilton's principle $\delta S = 0$ yields the following Euler-Poincar\'e equation
\begin{align}
    \left(\p_t + \mcal{L}_{\ob{u}_t}\right)  \frac{\delta \ell}{\delta \ob{u}_t} = -d p\,, \quad \text{where}\quad \frac{\delta \ell}{\delta \ob{u}_t} &= \ob{u}_t^\flat - \epsilon^2 \operatorname{div}_\nabla (\nabla \ob{u}^\flat_t \cdot F_t) + \epsilon^2 R(\ob{u}_t,\cdot):F_t \,. \label{eq:aniso alpha EP}
\end{align}
Here, $d: \Lambda^k(\clM) \rightarrow \Lambda^{k+1}(\clM)$ is the exterior derivative, $\operatorname{div}_\nabla(\nabla \ob{u}_t^\flat \cdot F_t) := \nabla_j(\nabla_k \ob{u}_t^\flat F^{jk}_t)$, and the pressure contribution arises from the equivalence class of the dual space to the space divergence free vector fields. The Lie derivative $\mathcal{L}_{\ob{u}}$ acting on $m \in \Lambda^1(\clM)$ can be expressed using covariant derivatives as
\begin{align*}
    \mcal{L}_{\ob{u}} m = \nabla_{\ob{u}} m + (\nabla_{\ob{u}})^T m\,.
\end{align*}
Together with the evolution of $F$ derived from the stochastic Taylor hypothesis \ref{assump:STH}, we have the closed system of equations 
\begin{align}
    \begin{split}
        &\p_t \frac{\delta \ell}{\delta \ob{u}_t} + \nabla_{\ob{u}_t} \frac{\delta \ell}{\delta \ob{u}_t} + \left(\nabla_{\ob{u}_t}\right)^T \frac{\delta \ell}{\delta \ob{u}_t} = -d p\,,\\
        &\p_t F_t + \nabla F_t \cdot \ob{u}_t - \nabla \ob{u}_t \cdot F_t - (\nabla \ob{u})^T \cdot F_t = (\xi \otimes \xi) \,,\\
        &\operatorname{div}_{\nabla} (\ob{u}_t) = \operatorname{div}_{\nabla} (\xi_i) = 0\,,
    \end{split}
\end{align}
which are the forced anisotropic Lagrangian averaged Euler equations. Several simplification can be made by specialising to flat geometries such that curvature $R$ vanishes. In this case, we have the following proposition.
\begin{proposition}
    Assuming the correlation tensor $F_t$ has the dynamics
    \begin{align}
        \p_t F_t + \mcal{L}_{\ob{u}_t} F_t = G_t\,,
    \end{align}
    for arbitrary, time dependent $G_t \in \mathfrak{X}(\clM)\otimes\mathfrak{X}(\clM)$ which can have functional dependence on $\ob{u}_t$, $w_t$ and $F_t$. Denoting by $C(q) = \operatorname{div}_{\nabla}(\nabla q \cdot F_t)$ for all arbitrary vector $q_t \in \mathfrak{X}(\clM)$, we have the commutation relation
    \begin{align}
        \left[\p_t + \nabla_{\ob{u}_t}, C \right] q_t = \operatorname{div}_{\nabla}(\nabla_{q_t} G_t)\,.
    \end{align}
\end{proposition}
\begin{proof}
    In the case $\p_t \wt{F}_t + \mcal{L}_{\ob{u}_t} \wt{F}_t = 0$, it was shown in e.g. \cite{Shkoller2002,MS2003}, that $\left[\p_t + \ob{u}_t \cdot \nabla, \wt{C} \right] q_t = 0$ for $\wt{C}u = \operatorname{div}_{\nabla^E}(\nabla u \cdot \wt{F}_t)$ for Euclidean connection $\nabla^E$. Since the contribution from $G_t$ enters only in term $\p_t \left(C (q_t)\right) = \p_t \operatorname{div}_{\nabla}(\nabla q_t \cdot F_t)$, we have the required result through a similar direct computation and switching of covariant derivatives, which always holds under the assumption that $R \equiv 0$.
\end{proof}
\begin{corollary}
    When the curvature $R$ vanishes, the Euler-Poincar\'e equation can \eqref{eq:aniso alpha EP} be expressed equivalently as
    \begin{align}
        (1 - \epsilon^2 C)\left(\p_t + \nabla_{\ob{u}_t}\right)\ob{u}_t - \epsilon^2 C(u)_i\nabla u_i = - \nabla \wt{p} - \operatorname{div}_{\nabla}(\nabla_{\ob{u}_t} (\xi \otimes \xi))\,,
    \end{align}
    where $\wt{p} = p + \frac{1}{2} \nabla(|u_t|^2)$ is the modified pressure.
\end{corollary}

\section{Conclusion}
In this work, we considered an averaged Euler-Poincar\'e variational principle to construct sub-grid scale ideal fluid models. In this construction, the total fluid flow was assumed stochastic and averaged through the expectation in analogy with the GLM theory. The first step of this construction, made in section \ref{sec:expand}, is the formulation of the mean velocity, defined by the factorisation of the total flow map into the fluctuation and mean maps. Then, the Magnus expansion of the stochastic flow map is considered so that the expectation operator is applied to vector fields and their actions. Our averages then occurred only at the tangent space of the identity, bypassing the difficulty in defining the expectation of diffeomorphism group-valued quantities. 
In section \ref{sec:vp}, we generalised the KE Lagrangian functional such that it is well-defined when the input velocity contains white noise terms using tools from Malliavin calculus. After making dynamical assumptions on the fluctuations vector fields in assumptions \ref{assump:STH} and \ref{assump:orthg}, we explicitly computed the averaged Lagrangian and derived a stochastically averaged version of the anisotropic Lagrangian averaged Euler's equation through an Euler-Poincar\'e variational principle.

Following this work, several open problems present themselves to be addressed in the future. First, the closure dynamics of $w_t$ and $\chi_t$ can be modified to include additional physics. For example, one can assume transport type stochastic perturbation for the dynamics of $w_t$, $\dd w_t + \ad_{w_t} \ob{u}_t + \ad_{w_t} \xi \circ \dd W_t = 0$ as to introduce dissipation into the dynamics of the correlation tensor $F_t$. 
Second, the driving stochastic processes can be generalised from Brownian motions to geometric rough paths. For example, fractional Brownian motion or lifts of Brownian motion that incorporate a noise-induced drift. Such drifts are natural in the context of multiscale limits \cite{CGH2017}. Lastly, the same modelling procedure can be generalised from incompressible Euler's fluid using $\operatorname{SDiff}(\mathcal{M})$ to semi-direct products such that sub-grid scale models for a large class of fluid flows can be derived.

\begin{credits}
\subsubsection{\ackname} We wish to thank D. Holm, O. Street and J. Woodfield for several thoughtful suggestions during the course of this work. TD is supported during the present work by European Research Council (ERC) Synergy grant Stochastic Transport in Upper Ocean Dynamics (STUOD) -- DLV-856408. RH is grateful for the support by Office of Naval Research (ONR) grant award N00014-22-1-2082, Stochastic Parameterization of Ocean Turbulence for Observational Networks. 

\subsubsection{\discintname}
The authors have no competing interests to declare that are
relevant to the content of this article.
\end{credits}
\bibliographystyle{splncs04}
\bibliography{main}

\begin{thebibliography}{10}
\providecommand{\url}[1]{\texttt{#1}}
\providecommand{\urlprefix}{URL }
\providecommand{\doi}[1]{https://doi.org/#1}

\bibitem{andrews_mcintyre_1978}
Andrews, D.G., Mcintyre, M.E.: {An exact theory of nonlinear waves on a
  Lagrangian-mean flow}. Journal of Fluid Mechanics  \textbf{89}(4),  609–646
  (1978). \doi{10.1017/S0022112078002773}

\bibitem{ACC2014}
Arnaudon, M., Chen, X., Cruzeiro, A.B.: {Stochastic {E}uler-{P}oincaré
  reduction}. Journal of Mathematical Physics  \textbf{55} (2014).
  \doi{10.1063/1.4893357}

\bibitem{BLANES2009151}
Blanes, S., Casas, F., Oteo, J., Ros, J.: The magnus expansion and some of its
  applications. Physics Reports  \textbf{470}(5),  151--238 (2009).
  \doi{https://doi.org/10.1016/j.physrep.2008.11.001}

\bibitem{CCR2023}
Chen, X., Cruzeiro, A.B., Ratiu, T.S.: {Stochastic Variational Principles for
  Dissipative Equations with Advected Quantities}. Journal of Nonlinear Science
   \textbf{33}, ~5 (2 2023). \doi{10.1007/s00332-022-09846-1}

\bibitem{CC2007}
Cipriano, F., Cruzeiro, A.B.: {Navier-Stokes equation and diffusions on the
  group of homeomorphisms of the torus}. Communications in Mathematical Physics
   \textbf{275},  255--269 (10 2007). \doi{10.1007/s00220-007-0306-3}

\bibitem{CGH2017}
Cotter, C.J., Gottwald, G.A., Holm, D.D.: {Stochastic partial differential
  fluid equations as a diffusive limit of deterministic Lagrangian multi-time
  dynamics}. Proceedings of the Royal Society A: Mathematical, Physical and
  Engineering Sciences  \textbf{473}(2205),  20170388 (2017).
  \doi{10.1098/rspa.2017.0388}

\bibitem{DHP2023}
Diamantakis, T., Holm, D.D., Pavliotis, G.A.: {Variational Principles on
  Geometric Rough Paths and the Lévy Area Correction}. SIAM Journal on Applied
  Dynamical Systems  \textbf{22},  1182--1218 (6 2023).
  \doi{10.1137/22M1522164}

\bibitem{DHL2024}
Diamantakis, T., Hu, R.: Variational closures for composite homogenised fluid
  flows (2024), \url{https://arxiv.org/abs/2409.10408}

\bibitem{FHT2001}
Foias, C., Holm, D.D., Titi, E.S.: The navier–stokes-alpha model of fluid
  turbulence. Physica D: Nonlinear Phenomena  \textbf{152-153},  505--519
  (2001). \doi{10.1016/S0167-2789(01)00191-9}

\bibitem{GJAJA1996343}
Gjaja, I., Holm, D.D.: Self-consistent hamiltonian dynamics of wave mean-flow
  interaction for a rotating stratified incompressible fluid. Physica D:
  Nonlinear Phenomena  \textbf{98}(2),  343--378 (1996).
  \doi{10.1016/0167-2789(96)00104-2}

\bibitem{Hall_2003}
Hall, B.C.: Lie Groups, Lie Algebras, and Representations. Springer New York
  (2003). \doi{10.1007/978-0-387-21554-9}

\bibitem{HH2021a}
Holm, D.D., Hu, R.: Stochastic effects of waves on currents in the ocean mixed
  layer. Journal of Mathematical Physics  \textbf{62} (2021).
  \doi{10.1063/5.0045010}

\bibitem{Holm2015}
Holm, D.D.: {Variational principles for stochastic fluid dynamics}. Proceedings
  of the Royal Society A: Mathematical, Physical and Engineering Sciences
  \textbf{471}(2176),  20140963 (2015)

\bibitem{https://doi.org/10.1002/cpa.3160070404}
Magnus, W.: On the exponential solution of differential equations for a linear
  operator. Communications on Pure and Applied Mathematics  \textbf{7}(4),
  649--673 (1954). \doi{https://doi.org/10.1002/cpa.3160070404}

\bibitem{MS2003}
Marsden, J.E., Shkoller, S.: The anisotropic lagrangian averaged euler and
  navier-stokes equations. Archive for Rational Mechanics and Analysis
  \textbf{166},  27--46 (1 2003). \doi{10.1007/s00205-002-0207-8}

\bibitem{Memin2014}
Mémin, E.: {Fluid flow dynamics under location uncertainty}. Geophysical \&
  Astrophysical Fluid Dynamics  \textbf{108},  119--146 (3 2014).
  \doi{10.1080/03091929.2013.836190}

\bibitem{di_nunno_malliavin_2009}
Nunno, G.D., Øksendal, B., Proske, F.: Malliavin Calculus for Lévy Processes
  with Applications to Finance. Springer Berlin Heidelberg, Berlin, Heidelberg
  (2009). \doi{10.1007/978-3-540-78572-9}

\bibitem{Shkoller2002}
Shkoller, S.: The Lagrangian Averaged Euler (LAE-$\alpha$) Equations with
  Free-Slip or Mixed Boundary Conditions, pp. 113--165. Springer-Verlag (2002).
  \doi{10.1007/0-387-21791-6_5}

\end{thebibliography}

\end{document}